\documentclass[a4paper,UKenglish,cleveref, autoref, thm-restate]{lipics-v2021}

\title{Geometry Matters in Planar Storyplans} %

\titlerunning{Geometry Matters in Planar Storyplans} %

\author{Alexander Dobler}{TU Wien, Vienna, Austria}{adobler@ac.tuwien.ac.at}{https://orcid.org/0000-0002-0712-9726}{Vienna Science and Technology Fund (WWTF)  grant [10.47379/ICT19035]}
\author{Maximilian Holzmüller}{TU Wien, Vienna, Austria}{}{}{}
\author{Martin Nöllenburg}{TU Wien, Vienna, Austria}{noellenburg@ac.tuwien.ac.at}{https://orcid.org/0000-0003-0454-3937}{Vienna Science and Technology Fund (WWTF)  grant [10.47379/ICT19035]}

\authorrunning{M.~Holzmüller, A.~Dobler and M.~Nöllenburg} %

\Copyright{John Q. Public and Joan R. Public} %

\ccsdesc[300]{Theory of computation~Computational geometry}
\ccsdesc[300]{Mathematics of computing~Graph theory}
\ccsdesc[300]{Human-centered computing~Graph drawings}

\keywords{geometric storyplan, planarity, straight-line drawing, dynamic graph drawing} %

\relatedversion{} %

\nolinenumbers %

\hideLIPIcs  %

\EventEditors{John Q. Open and Joan R. Access}
\EventNoEds{2}
\EventLongTitle{42nd Conference on Very Important Topics (CVIT 2016)}
\EventShortTitle{CVIT 2016}
\EventAcronym{CVIT} 
\EventYear{2016}
\EventDate{December 24--27, 2016}
\EventLocation{Little Whinging, United Kingdom}
\EventLogo{}
\SeriesVolume{42}
\ArticleNo{23}

\usepackage{amssymb,amsthm}
\usepackage{printlen,cite}
\usepackage[basic]{complexity}

\newtheorem{problem}{Problem}
\newcommand{\probname}[1]{{\normalfont\textsc{#1}}}

\begin{document}
\maketitle

\begin{abstract}
A \emph{storyplan} visualizes a graph $G=(V,E)$ as a sequence of $\ell$ frames $\Gamma_1, \dots, \Gamma_\ell$, each of which is a drawing of the induced subgraph $G[V_i]$ of a vertex subset $V_i \subseteq V$. 
Moreover, each vertex $v \in V$ is contained in a single consecutive sequence of frames $\Gamma_i, \dots, \Gamma_j$, all vertices and edges contained in consecutive frames are drawn identically, and the union of all frames is a drawing of $G$. 
In GD 2022, the concept of \emph{planar storyplans} was introduced, in which each frame must be a planar (topological) drawing. Several (parameterized) complexity results for recognizing graphs that admit a planar storyplan were provided, including \NP-hardness.
In this paper, we investigate an open question posed in the GD paper and show that the geometric and topological settings of the planar storyplan problem differ:
We provide an instance of a graph that admits a planar storyplan, but no planar \emph{geometric} storyplan, in which each frame is a planar straight-line drawing.
Still, by adapting the reduction proof from the topological to the geometric setting, we show that recognizing the graphs that admit planar geometric storyplans remains \NP-hard. 
\end{abstract}

\section{Introduction}
Planar storyplans, introduced in 2022~\cite{BinucciGLLMNS22}, represent an approach to draw possibly dense and non-planar graphs as a sequence of $\ell \ge 1$ planar subdrawings called frames. 
In addition to the planarity condition for each frame, their sequence must satisfy certain visual consistency criteria in order to be a valid planar storyplan: 
(i) each vertex belongs to a single, non-empty interval of frames, (ii) each frame shows the subgraph induced by the vertices of the frame, (iii) the shared vertices and edges of any two consecutive frames have the same geometric representation, and (iv) every edge must be visible in at least one frame, i.e., the incident vertices of every edge must co-occur in some frame. 
Storyplans belong to the class of gradual graph visualizations, which depict a graph in a sequential story-like or unordered small-multiples fashion. 
The goal of such gradual visualizations is to show a complex graph $G$ as a collection of simpler drawings, e.g., planar drawings of subgraphs whose union represents $G$. 
Identical vertices and edges in the individual subdrawings can be mentally linked by requiring that they are always drawn the same if they occur in multiple subdrawings. 
Such visualizations %
find applications both for visualizing dynamic graphs~\cite{BeckBDW14,VehlowBW16}, in which vertices and edges occur over time, and for graphs that are decomposed into multiple simpler subgraphs without a sequential order~\cite{GiacomoDLMT13,GiacomoDLMT14}. 
Storyplans with their temporal sequence of frames are related to drawing graphs in a streaming model. In the storyplan problem, however, one can choose the sequence in which vertices enter and leave the story, whereas this sequence is part of the input in streaming models~\cite{LozzoR19,BinucciBBDGPPSZ09,BinucciBBDGPPSZ12,GoodrichP13a} and in graph stories~\cite{BattistaDGGOPT22,BorrazzoLFP19,BorrazzoLBFP20,BattistaDGGOPT23}.

The authors of~\cite{BinucciGLLMNS22,BinucciGLLMNS24} investigated the complexity of deciding whether a given graph admits a planar storyplan in a topological setting, with edges drawn as simple curves between their endpoints. 
They showed \NP-completeness and fixed-parameter tractability parameterized by the vertex cover number or the feedback edge set number of the graph. 
Furthermore, they proved that every graph of treewidth at most 3 admits a planar storyplan, and that such a planar storyplan (actually with straight-line edges) can be computed in linear time. 
In~\cite{FialaFLWZ24} the investigation of storyplans was extended to outerplanar and forest storyplans, in which each frame is not just planar, but actually outerplanar or a forest. 
They proved a strict containment relation between forest, outerplanar, and planar storyplans and identified graph families that do or do not always admit these storyplans.
In the affirmative case, their storyplans use straight-line edges. 

In this paper, we investigate an open question by several authors~\cite{BinucciGLLMNS22,BinucciGLLMNS24} and explore the differences between planar (topological) storyplans and planar geometric storyplans.
Obviously, every graph that admits a planar geometric storyplan also admits a planar storyplan. 
We prove that the converse is not true by providing a counterexample of a graph with 28 vertices that admits a planar storyplan, but none with straight-line edges. 
The main challenge in our geometric construction is to enforce an obstruction to straight-line visibility for any possible placement of the vertices and for any possible vertex-to-frame assignment.
As a second result, we adapt the hardness proof of from~\cite{BinucciGLLMNS22,BinucciGLLMNS24} to establish that deciding if a given graph admits a planar geometric storyplan remains \NP-hard.

\section{Preliminaries}
For $n\in \mathbb{N}$, we use $[n]=\{1,2,\dots,n\}$ as a shorthand. 
For $a\le b\in \mathbb{N}$, we define $[a,b]=\{a,a+1,\dots,b\}$. 
For a graph $G$, let $V(G)$ be its vertex set and let $E(G)$ be its edge set. 
For $V'\subseteq V(G)$, $G[V']$ is the subgraph of $G$ induced by $V'$. 
Given an induced subgraph $G'$ of some graph $G$, and a vertex of $G$ not contained in $V(G')$, $G'+v$ is defined as $G[V(G')\cup \{v\}]$. 
We denote by $uv$ the edge $\{u,v\}$ connecting vertices $u$ and $v$.

\subparagraph*{The storyplan problem.}
A \emph{storyplan} of a graph $G$ on time steps~$[\ell]$, $\ell\in \mathbb{N}$, consists of a pair $(A,\mathcal{D})$. For each $v\in V(G)$, $A(v)=[s_v,e_v]\subseteq [\ell]$ is the \emph{visible interval} of $v$. We say that $v$ \emph{appears} at $s_v$, is \emph{visible at} each $s_v,s_v+1,\dots,e_v$, and \emph{disappears} after~$e_v$.
It must hold for a pair of adjacent vertices $u,v\in V$ that $A(u)\cap A(v)\ne \emptyset$, i.e., there is a time step in which the edge $uv$ can be drawn. We say in this case that $u$ and $v$ \emph{co-occur}.
For each $t\in [\ell]$, $G[t]=G[\{v\in V \mid t\in A(v)\}]$ is the \emph{frame graph at $t$}. For a subgraph $G'$ of $G$, we say that $G'$ is \emph{visible at $t$} if all of its vertices are visible.
The \emph{drawing function} $\mathcal{D}$ represents a drawing $\mathcal{D}(v)$ of each vertex $v\in V(G)$ and drawing $\mathcal{D}(e)$ of each edge $e\in E(G)$. Further, for each subgraph $G'$ of $G$, $\mathcal{D}(G')$ is the drawing of $G'$, i.e., the drawing function is applied to each vertex and edge of $G'$. For each $t\in [\ell]$, we call $\Gamma_t:=\mathcal{D}(G[t])$ the \emph{frame at $t$}.
For a storyplan, it must hold that each frame represents a simple drawing of the corresponding induced subgraph.
A storyplan is \emph{planar} if each frame corresponds to a planar drawing.
A storyplan is \emph{geometric} (or \emph{straight-line}) planar if each frame corresponds to a planar drawing such that each edge is represented by a straight-line segment.

We consider the following two problems.
\begin{problem}\label{problem:planar}
    Given a graph $G$, does $G$ admit a planar storyplan?
\end{problem}

\begin{problem}\label{problem:geomplanar}
    Given a graph $G$, does $G$ admit a planar geometric storyplan?
\end{problem}

Note that these definitions are slightly different from those in~\cite{BinucciGLLMNS24}. There, each vertex appears in a unique frame, while we permit multiple vertices appearing in a single frame. It is easy to see, though, that these two definitions are equivalent with regard to the existence of a planar (geometric) storyplan for a given graph $G$.

\section{Geometry Matters in Planar Storyplans}
In this section, we show our following main result.
\begin{theorem}\label{thm:mainthm}
    There is a graph that admits a planar storyplan and that does not admit a planar geometric storyplan.
\end{theorem}
The graph $G$ is defined as follows. 
It consists of three four-cycles $A$, $B$, and $C$ with vertices $a_i,b_i,c_i$, $i\in [4]$ and edges $e_i^A = a_ia_{i\bmod{4}+1}$, $e_i^B= b_ib_{i\bmod{4}+1}$, and $e_i^C = c_ic_{i\bmod{4}+1}$ for~$i \in [4]$.
Furthermore, the graph contains eight \emph{apex vertices} $q_i^j$, $i\in [4], j\in [2]$, such that each $q_i^j$ is connected to all vertices of $A$, $B$, and $C$.
Lastly, eight more \emph{edge vertices} $r_i^j$, $i\in [4], j\in [2]$ are added. Each $r_i^j$ is connected to $q_i^j$ and to all vertices from $e_i^A,e_i^B$, and $e_i^C$. 

We define $Q=\{q_i^j\mid i\in [4],j\in [2]\}$ and $R=\{r_i^j\mid i\in [4], j\in [2]\}$. For each $i\in [4]$ and $j\in [2]$, $P_i^j=\{q_i^j,r_i^j\}$ is an \emph{apex-pair}.

Below, we give a formal definition of $G$. Further, the graph can also be parsed from one of its planar storyplans in \cref{fig:storyplanforg}.
\begin{align*}
    V(G)=&\{a_i,b_i,c_i\mid i\in [4]\}\cup \{q_i^j,r_i^j\mid i\in [4],j\in [2]\}\\
    E(G)=&\{a_ia_{i\bmod{4}+1},b_ib_{i\bmod{4}+1},c_ic_{i\bmod{4}+1}\mid i\in [4]\} \cup \{r_i^jq_i^j\mid i\in [4],j\in [2]\} \cup \\
    &\{r_i^ja_i,r_i^ja_{i\bmod{4}+1},r_ib_i,r_ib_{i\bmod{4}+1},r_i^jc_i,r_i^jc_{i\bmod{4}+1}\mid i\in [4],j\in [2]\}\cup \\
    &\{q_i^ja_k,q_i^jb_k,q_i^jc_k\mid i,k\in [4], j\in [2]\}
\end{align*}
We prove both statements of \cref{thm:mainthm} in the following two sections.

\subsection{\texorpdfstring{$\boldsymbol{G}$}{} Has a Planar Storyplan}
\begin{figure}[htb!]
    \centering
    \includegraphics[page=1]{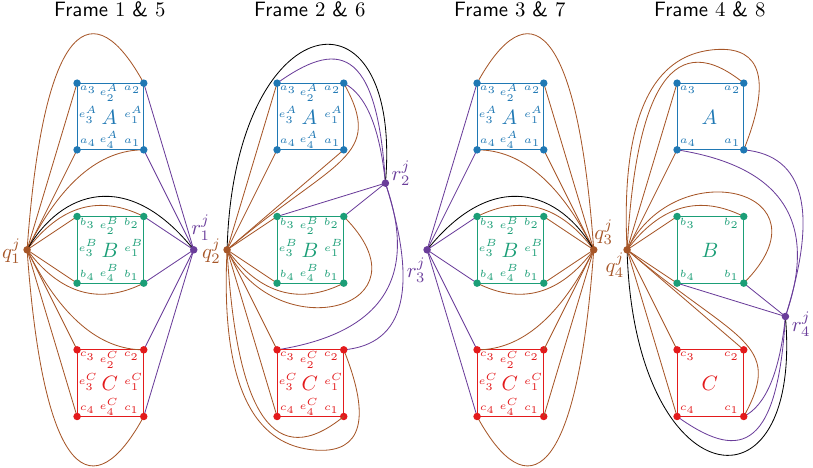}
    \caption{A planar storyplan for the graph $G$ with eight frames. 
    Frames $k$ and $k+4$ show the same drawing for each $k\in [4]$. For frames $1,2,3,4$, $j=1$; for frames $5,6,7,8$, $j=2$.}
    \label{fig:storyplanforg}
\end{figure}
A planar storyplan of $G$ with eight frames $\Gamma_1, \dots, \Gamma_8$ is given in \cref{fig:storyplanforg}. The four-cycles~$A$, $B$, $C$ are visible in all frames. Each apex-pair $q_i^j,r_i^j$ for $i \in [4]$ and $j\in [2]$ is visible only in frame $\Gamma_{(j-1)4+i}$. As the neighborhoods of apex pairs $q_i^1,r_i^1$ and $q_i^2,r_i^2$, $i\in [4]$, w.r.t.\ the four-cycles are the same, they can be represented by the same drawing. This is why we show only four drawings in the figure.

\subsection{\texorpdfstring{$\boldsymbol{G}$}{} Does Not Have a Planar Geometric Storyplan}
This section serves the purpose of showing the following.
\begin{theorem}\label{thm:doesnotadmit}
    The graph $G$ does not admit a  planar geometric storyplan.
\end{theorem}
We proceed to assume that $G$ admits a planar geometric storyplan $(A,\mathcal{D})$, and we will arrive at a contradiction.
We require a sequence of lemmas that will lead to this contradiction.
\begin{lemma}\label{lemma:apexorder}
    Let $I=\{A(q_i^j)\mid q_i^j\in Q\}$ be the visible intervals of the vertices in $Q$. Then there exists a set $I'\subset I$ of cardinality at least 6, such that for all $[s,e]\in I'$ and all $t\in [s,e]$, $A$, $B$, and $C$ are visible at $t$. 
\end{lemma}
\begin{proof}
    First, there is a time step in which all cycle vertices from a single cycle, say $A$, are visible. Otherwise, consider the time step $e_A$ after which the first cycle vertex of $A$, say $a_1$, disappears, yet $a_3$ has not yet appeared. Then, in $e_A$, the three cycle vertices $a_1, a_2, a_4$ of $A$ and all eight apex vertices must be visible since they need to see both $a_1$ and $a_3$; they induce a non-planar $K_{3,8}$, contradicting that every frame is planar.
    This holds analogously for the cycles $B,C$.
    Secondly, there is a time step in which all cycle vertices of all three cycles are visible. Otherwise, consider the time step after which the first cycle vertex disappears. Again, at this time step, we have at least a non-planar $K_{4,8}$ by the same argumentation that the apex vertices need to see both the disappearing cycle vertex and at least one cycle vertex that has not yet appeared. 
    
    Now sort the intervals in $I$ by their start points.
    Consider the first two intervals $[s,e],[s',e']$ in this order. Let $q$ and $q'$ be their respective apex-vertices.
    We show that $[s,e]\cap [s',e']=\emptyset$. By our previous arguments, there is a time step $t\in [s,e]$ in which $q$ and all cycle vertices are visible.  If $s\le s'\le e$, then there is furthermore a time step $t'$ with $s'\le t'\le e$, such that~$q$, $q'$, and all cycle vertices are visible. Notice that this amounts to $14$ vertices and $36$ edges. Further, the induced drawing of the four-cycles and the two apex-vertices must clearly have a face incident to four vertices -- four vertices of one of the four-cycles. Adding a chord, keeping the structure planar, we would end up with $14$ vertices and $37$ edges, a contradiction to Euler's formula.
    Symmetrically, by considering the last two intervals when sorted by their end points, we also see that their intersection is empty. The statement follows directly, as each apex must co-occur with each cycle-vertex.
\end{proof}

\begin{corollary}\label{corr:apexpairspluscycles}
    There exist at least 6 frames, each containing a different apex pair and all four-cycles.
\end{corollary}

\begin{lemma}
    There exists a pair from $\{A, B, C\}$ that is not nested, i.e., in the drawing of only these two  four-cycles, the outer face is incident to all vertices.
\end{lemma}
\begin{proof}
    This follows directly from \cref{lemma:apexorder}, as there exists a frame with some apex vertex and all three four-cycles. Clearly, the statement holds, as there is otherwise no face (in the induced drawing of $A$, $B$, and $C$) that contains all four-cycle vertices.
\end{proof}
In the remainder of the proof, we assume that $A$ and $B$ form a non-nested pair.

Next, we want to define one of the drawings of $A$ or $B$ to be ``smaller'', see \cref{fig:defsmaller}. This will be useful in showing that for one of the four-cycles $A$ and $B$, in the drawing of this four-cycle plus an apex-vertex, the other four-cycle will always lie on the outer face.
\begin{figure}
    \centering
    \includegraphics[page=2]{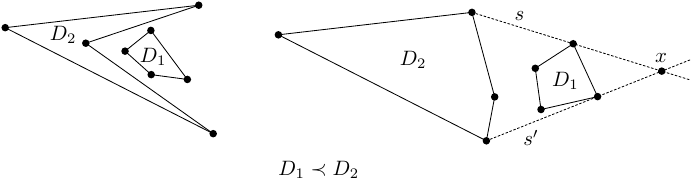}
    \caption{The two cases of a $D_1\prec D_2$. In the second case, the point $x$ is closer to $D_1$ than to $D_2$.}
    \label{fig:defsmaller}
\end{figure}
\begin{definition}\label{def:smaller}
    Consider two simple quadrilaterals $D_1$ and $D_2$, such that the outer face is incident to all eight vertices.
    Define $D_1\prec D_2$ if one of the following conditions holds.
    \begin{enumerate}
        \item $D_1$ is contained within the convex hull of $D_2$.
        \item The convex hull of $D_1\cup D_2$ contains two non-parallel line segments $s$ and $s'$ each connecting a point from $D_1$ with a point from $D_2$. Let $x$ be the intersection point of the supporting lines of $s$ and $s'$. Point $x$ is closer to $D_1$ than to $D_2$ (in terms of the classic definition between a point and the set of points defined by the inside region of the quadrilateral).
    \end{enumerate}
\end{definition}
It is clear that we cannot have $D_1\prec D_2$ and $D_2\prec D_1$.
It can, however, be that we have neither. %
As the drawings of $A$, $B$, and $C$ correspond to simple quadrilaterals, we can show the following lemma (refer to \cref{fig:implication} for an illustration). Note that the statement makes sense, as $A$ and $B$ share the outer face.
\begin{figure}
    \centering
    \includegraphics[page=3]{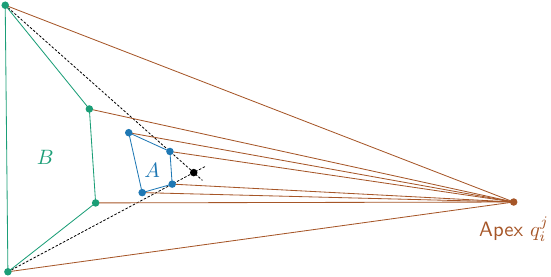}
    \caption{$\mathcal{D}(A)\prec \mathcal{D}(B)$ is implied by $A$ lying inside of one of the faces of the drawing of $A+q_i^j$.}
    \label{fig:implication}
\end{figure}
\begin{lemma}
    Let $F$ be a frame containing an apex pair $P_i^j$ and all four-cycles. 
    We have that $\mathcal{D}(A)\prec \mathcal{D}(B)$ if $\mathcal{D}(A)$ is contained within an inner face of $\mathcal{D}(B+ q_i^j)$.
    Equivalently, $\mathcal{D}(B)\prec \mathcal{D}(A)$ if $\mathcal{D}(B)$ is contained within an inner face of $\mathcal{D}(A+ q_i^j)$.
\end{lemma}
\begin{proof}
    Assume that $\mathcal{D}(A)$ is contained within an inner face of $\mathcal{D}(B+ q_i^j)$.
    For the frame~$F$, consider the drawing induced by $q_i^j$, $A$, and $B$.
    Consider the edges on the outer face of this drawing. Two of these connect a vertex of $B$ with $q_i^j$; we denote these vertices as~$b$ and~$b'$.
    Now, either $A$ is contained within the convex hull of $B$, or the convex hull of $A+B$ contains two line segments $s$ and $s'$, each connecting a vertex of $A$ with a vertex of $B$. In the former case, $\mathcal{D}(A)\prec \mathcal{D}(B)$ is immediate. 
    In the latter case, $s$ and $s'$ are contained within the triangle defined by $b$, $b'$, and $q_i^j$. Now notice that, if the intersection point of the lines defined by $s$ and $s'$ was closer to $\mathcal{D}(B)$ than to $\mathcal{D}(A)$, or if $s$ and $s'$ were parallel, the points on $s$ and~$s'$ which are vertices of $B$ would not be visible from $q_i^j$. Hence, we have $\mathcal{D}(A)\prec \mathcal{D}(B)$.
\end{proof}
As we cannot have both $\mathcal{D}(A)\prec \mathcal{D}(B)$ and $\mathcal{D}(B)\prec \mathcal{D}(A)$, we assume, w.l.o.g., that $\mathcal{D}(B)\not\prec \mathcal{D}(A)$. Thus, in any frame containing $A$, $B$, and some apex vertex $q_i^j$, $B$ will lie on the outer face of the drawing defined by $A$ and $q_i^j$.

We will now show the last ingredient, which can be seen as the key lemma of the construction. Essentially, we want to show that at most two edges of $A$ are ``good'' -- in the sense that they can lie on the outer face of the drawing of $A$ plus some apex vertex (see \cref{fig:halfplanes}b). 
More formally, for a fixed drawing of $A$, the edge~$e_k^A$, $k\in [4]$, is \emph{good} if it can be the edge on the outer face of the planar straight-line drawing $\mathcal{D}(A+q_i^j)$ for some possible placement of some apex vertex $q_i^j$ in the outer face of $\mathcal{D}(A)$.
We show the following.
\begin{figure}
    \centering
    \includegraphics[page=4]{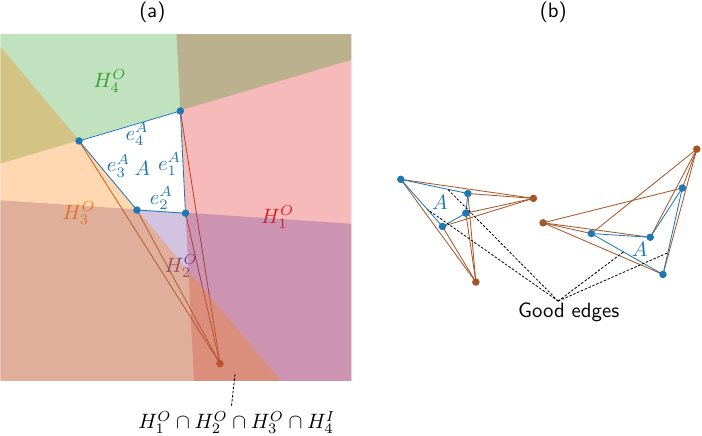}
    \caption{(a) The half-planes defined by $\mathcal{D}(A)$, an apex vertex must lie in a specific intersection of half-planes. (b) A convex and concave drawing of $A$ with exactly two good edges.}
    \label{fig:halfplanes}
\end{figure}
\begin{lemma}\label{lemma:twogoodsides}
    For a fixed drawing $\mathcal{D}(A)$ of $A$, $A$ has at most two good edges.
\end{lemma}
\begin{proof}
    The proof distinguishes two cases, depending on whether $\mathcal{D}(A)$ is convex or not. We start with the convex case.
    We show that it is not possible that $e_1^A$ and $e_3^A$ are both good. By symmetry, this means that $e_2^A$ and $e_4^A$ cannot be both good, and the statement follows.
    Thus, assume that $e_1^A$ and $e_3^A$ are good. We observe that in the drawing of $A$, each edge defines two open half-planes -- the ``outer half-plane'' and the ``inner half-plane'' (see \cref{fig:halfplanes}a for the outer half-planes).
    More precisely, the outer half-plane defined by $e_i^A$ contains as subset the points which can be connected by straight lines to $a_i$ and $a_{i\pmod 4+1}$ without crossings, while the inner half plane does not.
    Thus, for $k\in [4]$, define $H_k^O$ as the outer half-plane defined by~$e_k^A$, and $H_k^I$ as the inner half-plane. As the half-planes are open, they do not contain the line defined by their corresponding edge.
    Since $e_1$ is good, $H_1^I\cap H_2^O\cap H_3^O\cap H_4^O$ is non-empty (an apex vertex must be able to lie in this region for the drawing to be straight-line planar and for the edge $e_1^A$ to be on the outer face).  Similarly, as $e_3$ is good $H_1^O\cap H_2^O\cap H_3^I\cap H_4^O$ is non-empty.
    Note that $H_2^O\cap H_4^O$ is disjoint from the inner region defined by $\mathcal{D}(A)$. For~$H_1^O$ and $H_1^I$ to both share points with the intersection of~$H_2^O$ and $H_4^O$,
    the line defined by the drawing of $e_1^A$ must pass through $H_2^O\cap H_4^O$. Because the drawing of $A$ is convex, and because  $H_2^O\cap H_4^O$ does not share points with the region defined by $\mathcal{D}(A)$, this is not possible. We obtain a contradiction.

    For the concave case, w.l.o.g., assume that the concave corner is at $a_2$, with its two incident edges $e_1^A$ and $e_2^A$. 
    We claim that none of $e_1^A$ and $e_2^A$ can be good.
    For $e_1^A$ to be good, we again have that $H_1^I\cap H_2^O\cap H_3^O\cap H_4^O$ is non-empty. This is not possible: Consider any point $x$ outside $\mathcal{D}(A)$ and in $H_1^I$. Notice that the line segment from $x$ to $a_2$ intersects an edge of $A$. Hence, $e_1^A$ cannot be good. The proof that $e_2$ cannot be good proceeds equivalently.
\end{proof}
\begin{proof}[Proof of \cref{thm:doesnotadmit}]
    Assume for a contradiction that $G$ admits a planar geometric storyplan. 
    By our assumption, the drawings of $A$ and $B$ are in each other's outer face. Further $\mathcal{D}(B)\prec \mathcal{D}(A)$ does not hold. Thus, in each of the six frames existing by \cref{corr:apexpairspluscycles}, $\mathcal{D}(B)$ is in the outer face of the drawing induced by $A$ and the apex vertex. For the frame to be planar, the edge vertex of the apex-pair must be on the outer face of the drawing induced by the apex and $A$. Hence, the  edge vertex must be connected with a good side of $A$. But, as there are only two good sides (\cref{lemma:twogoodsides}), this is only possible for at most four apex-pairs.  Hence, two of the edge vertices of the six apex pairs of \cref{corr:apexpairspluscycles} cannot be connected to a good side of $A$, a contradiction.
\end{proof}

\section{NP-hardness}
Furthermore, we can show that 
\cref{problem:geomplanar} 
remains \NP-hard.
The proof is an adaptation of the hardness proof in \cite{BinucciGLLMNS24} for \cref{problem:planar}. In fact, the reduction is exactly the same. However, we have to adapt one direction of the correctness proof.

The reduction is from  \probname{One-In-3SAT}. \probname{One-In-3SAT} is an \NP-hard variant of \probname{3SAT}, where one asks for a satisfying assignment such that exactly one literal is true in each clause. Thus, let $\phi$ be such a formula over $n$ variables $x_1,\dots,x_n$ and $m$ clauses $C_1,\dots,C_m$. 
The graph $G$ consists of three types of gadgets: variable gadgets, clause gadgets, and wire gadgets (see \cref{fig:gadgets}). 

\begin{figure}[ht]
    \centering
    \includegraphics{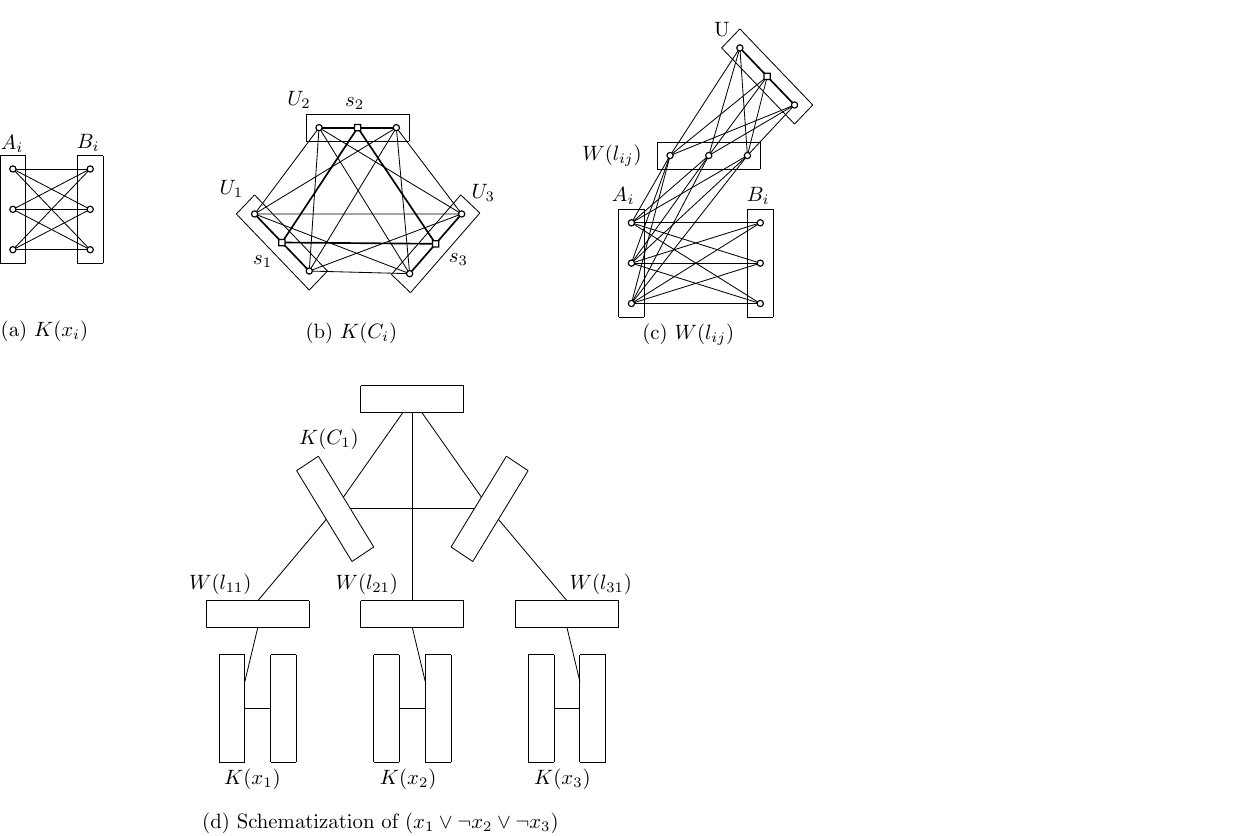}
    \caption{Illustration of the various types of gadgets and how they are connected.}
    \label{fig:gadgets}
\end{figure}

\subparagraph*{Variable gadget.} Each variable $x_i$ is represented in $G$ by a copy $K(x_i)$ of $K_{3,3}$ (see \cref{fig:gadgets}a). Let $A_i$ and $B_i$ be the two partite sets of $K(x_i)$ which we call the \emph{v-sides} of $K(x_i)$. 

\subparagraph*{Clause gadget.} Consider a copy of $K_{2,2,2}=(U_1\cup U_2\cup U_3,F)$. An extended $K_{2,2,2}$ is the graph obtained from any such copy by adding three vertices $s_1,s_2,s_3$, such that these three vertices are pairwise adjacent, and each $s_j$ is adjacent to both vertices in $U_j$, for $j\in [3]$. The vertices $s_1,s_2,s_3$ are called the \emph{special vertices}. Each clause $C_i$ is represented in $G$ by an extended $K_{2,2,2}$ denoted as $K(C_i)$ (see \cref{fig:gadgets}b).
We call each of the three sets of vertices $U_j\cup \{s_j\}$ a c-side of $K(C_i)$. For $j\in [3]$, we denote the two elements in $U_j$ by $u_{j1}$ and $u_{j2}$.

\subparagraph*{Wire gadget.} Refer to \cref{fig:gadgets}c. Let $x_i$ be a variable having a literal $\ell_{ij}$ in a clause~$C_j$, i.e.,~$x_i$ appears in $C_j$ positively or negatively. Any such variable-clause incidence is represented in $G$ by a set of three vertices, which we call the w-side $W(\ell_{i,j})$. All vertices of $W(\ell_{ij})$ are connected to all vertices of one of the three c-sides of $K(C_j)$, which we call $U$, such that the graph induced by $W({\ell_{ij}})\cup U$ in $G$ corresponds to a $K_{3,3}$. Also, each vertex of $W(\ell_{ij})$ is connected to all vertices of the v-side $A_i$ ($B_i$) if the literal $\ell_{ij}$ is positive (negative), again resulting in a $K_{3,3}$.

The resulting graph is the same as in the reduction in \cite{BinucciGLLMNS24}. There, it was shown that, if $G$ admits a planar storyplan, then $\phi$ has a satisfying assignment in which each clause contains exactly one true literal. As each geometric planar storyplan is a planar storyplan, we obtain the following.

\begin{corollary}[\hspace{-0.1pt}\cite{BinucciGLLMNS24}]\label{corr:oldirection}
    If $G$ admits a  planar geometric storyplan, then $\phi$ has a satisfying assignment in which each clause contains exactly one positive literal.
\end{corollary}
The other direction, however, does not follow from \cite{BinucciGLLMNS24}, as there, the storyplan constructed from a satisfying assignment contains non-straight edges.

Hence, we prove the following.
\begin{lemma}\label{lemma:newdirection}
    If $\phi$ has a satisfying assignment in which each clause contains exactly one positive literal, then $G$ admits a  planar geometric storyplan.
\end{lemma}
\begin{proof}
    Let $L$ consist of the literals that appear positively in the satisfying assignment. That is, for each clause $C_j$ in $\phi$, $L$ contains exactly one literal $\ell_{ij}$ for some $i\in [n]$.
    We provide a planar geometric storyplan.
    The drawing of the variable-, and wire-gadgets will essentially be the same as in \cite{BinucciGLLMNS24}. We will adapt the storyplan for the clause gadgets.
    First, we have a set of vertices $V_f$ that are visible in every frame of the storyplan. For each literal $\ell_{ij}\in L$, the corresponding vertices of $W(\ell_{ij})$ will always be visible. Further, for each variable $x_i$ that is set to true (false), the set $B_i$ ($A_i$) will always be visible. 
    In the first frame, all the vertices in $V_f$ will appear. Note that the corresponding induced graph does not contain any edges. 
    
    \begin{figure}
        \centering
        \includegraphics{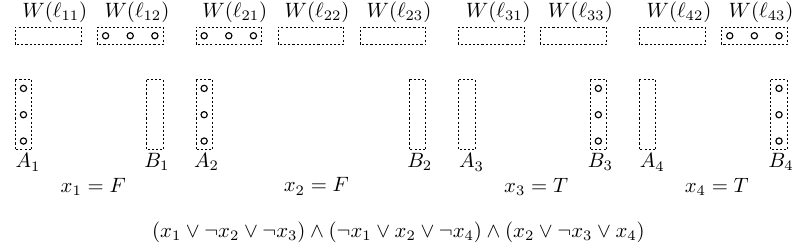}
        \caption{The drawing of the variable-, and wire-gadgets in the first frame.}
        \label{fig:vwgadgets}
    \end{figure}
    We proceed by showing how the vertices of the vertex-, and wire-gadgets are embedded.  See \cref{fig:vwgadgets} where the drawing of the first frame is shown, and the regions where the not-yet drawn vertices will appear for the variable- and wire-gadgets are sketched. The wire-gadget vertices are all drawn on a horizontal \emph{wire-line}, ordered based on the variables they are connected with. Hence, wire gadgets that connect to the variable $x_1$ are placed at the left, followed by wire gadgets for $x_2$, and so on. The vertices corresponding to a single wire gadget are consecutive. Further, the variable gadgets are drawn as in \cref{fig:vwgadgets}, each on two vertical lines facing each other, the $A$-vertices on the left, the $B$-vertices on the right. All wire gadgets for a variable $x_i$ are enclosed horizontally between $A_i$ and $B_i$.

    \begin{figure}
        \centering
        \includegraphics{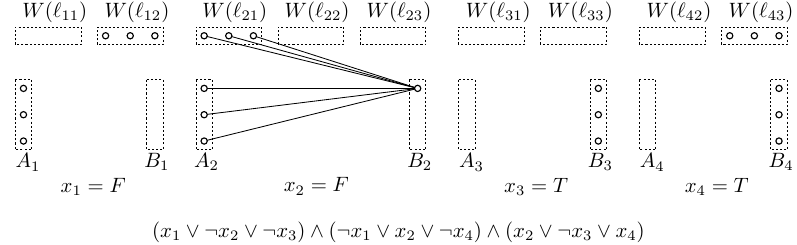}
        \caption{One of the frames drawing the remaining vertices of vertex-gadgets.}
        \label{fig:connectevw}
    \end{figure}
    In the next sequence of frames, we can draw all variable vertices that have not appeared yet. Hence, let $x_i$ be some true (false) variable, and let $v$ be some vertex of $A_i$ ($B_i$). We create a single frame in which $v$ appears. As $v$ has then co-occurred with all of its adjacent vertices in $G$, $v$ can disappear in the next frame again. Refer to \cref{fig:connectevw} for one such frame.

    \begin{figure}
        \centering
        \includegraphics[page=1]{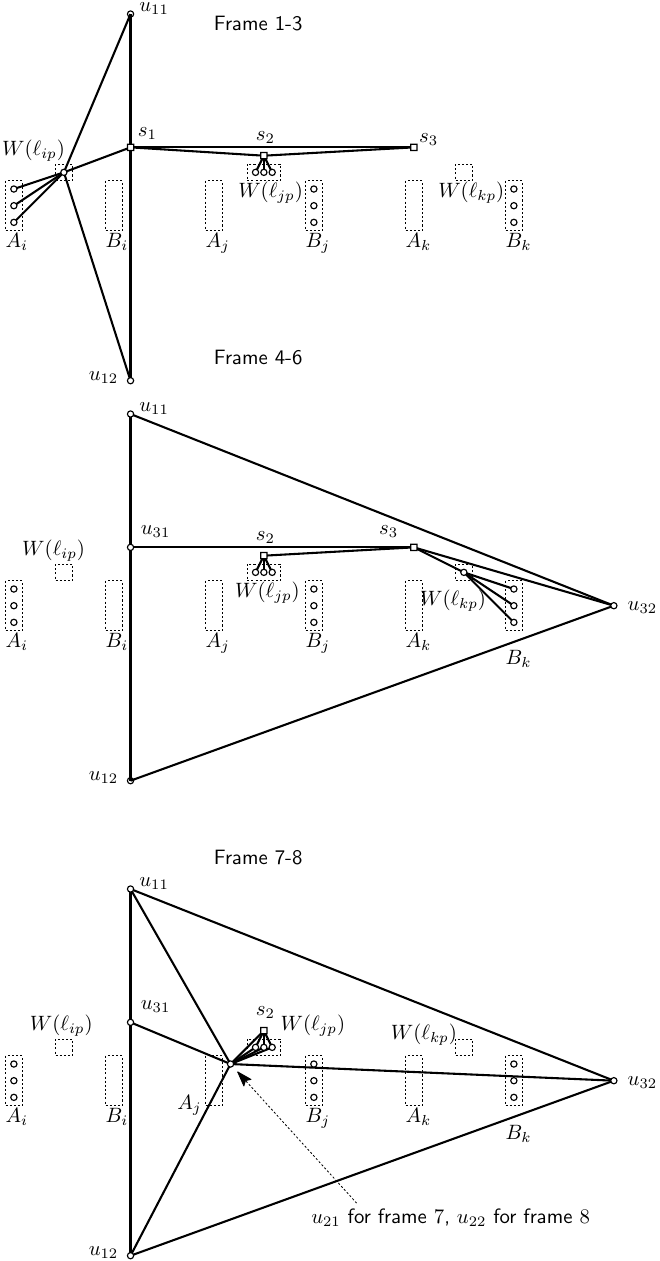}
        \caption{Drawing the clause-gadget if the true literal is in the middle.}
        \label{fig:clausegadgetmiddle}
    \end{figure}
    \begin{figure}
        \centering
        \includegraphics[page=2]{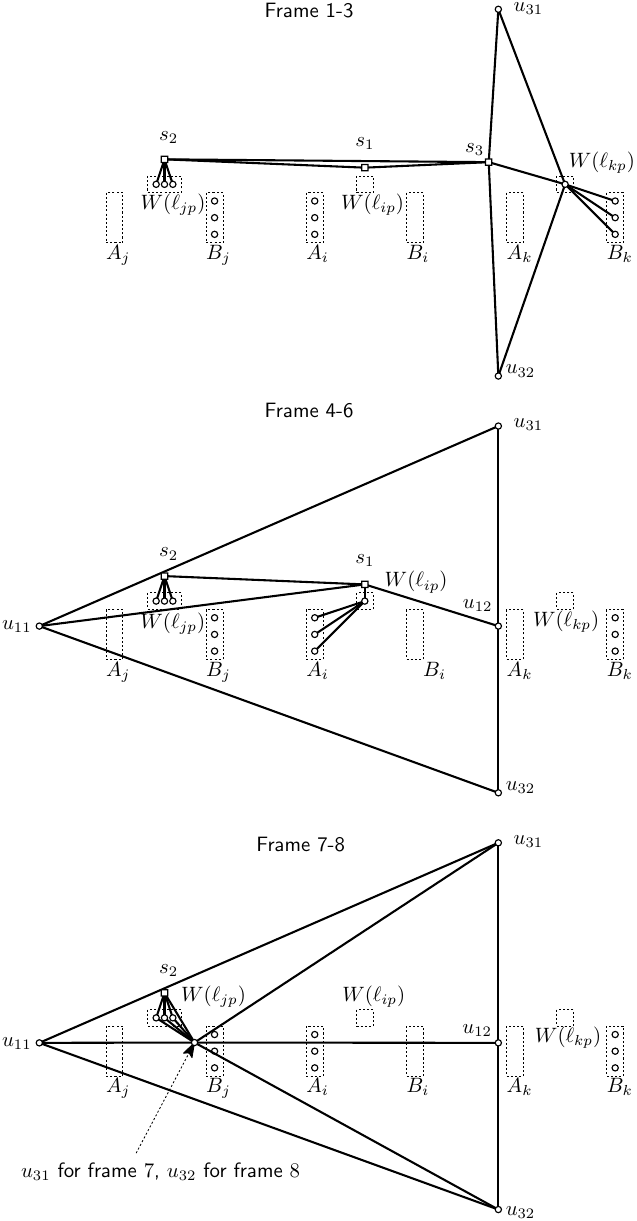}
        \caption{Drawing the clause-gadget if the true literal is on the left. If it is right, the frames can be mirrored.}
        \label{fig:clausegadgetleft}
    \end{figure}
    Until this point, the storyplan was the same as in \cite{BinucciGLLMNS24}. It remains to draw the clause gadgets and the missing wire gadgets. In \cite{BinucciGLLMNS24}, non-straight edges were used. Thus, we have to present a way of drawing the clause gadgets with straight lines. Consider a clause $C_p$ with literals $\ell_{ip},\ell_{jp},\ell_{kp}$, where $W(\ell_{ip})$ is connected to the first c-side of $K(C_p)$, $W(\ell_{jp})$ to the second, and $W(\ell_{kp})$ to the third.
    As an example, let us assume that $C_p=x_i\lor x_j\lor \lnot x_k$ and further assume that $x_i$ is false, and $x_j$ and $x_k$ are true. It will become clear that the construction also works for different clauses that use different combinations of positive and negative literals. Notice that $\ell_{jp}$ is the true literal. Now there are two cases: either $W(\ell_{jp})$ is between $W(\ell_{ip})$ and $W(\ell_{kp})$ on the wire-line, or it is left or right. For both cases, we present eight frames to draw the clause gadget for $C_p$ and the remaining wire vertices for $W(\ell_{ip})$ and $W(\ell_{kp})$ in \cref{fig:clausegadgetmiddle,fig:clausegadgetleft}, respectively.
    Note that in both figures, the drawn edges are the only visible edges in the corresponding frames.
    For both cases, in the first 3 frames all special vertices are drawn.     
    Further, one of the c-sides for a false literal is drawn -- $W(\ell_{ip})$ in \cref{fig:clausegadgetmiddle}, and $W(\ell_{ik})$ in \cref{fig:clausegadgetleft}. Each of the three frames shows a different wire-vertex for this literal. 
    After these three frames, the first special vertex can disappear and all the wire-vertices for the first false literals have co-occurred with all of their adjacent vertices in $G$. Next, we essentially do the same for the second false literal in frames 4--6. However, we have edges between the two $K_{2,2,2}$ vertices of the c-sides of the two false literals. We make sure that the four-cycle spanned by these vertices encloses the whole construction.
    The last two frames each represent one of the $K_{2,2,2}$-vertices of the c-side corresponding to the true literal. As these are not connected, we can have two frames in which exactly one of them appears and disappears again afterward. Hence, we have drawn all the vertices of the clause gadget and all their connections, and each vertex that appeared in one of the eight frames has co-occurred with all its adjacent vertices in $G$.
    Note that it does not matter in the construction how $C_p$ looks like with respect to its literals being positive or negative -- in the corresponding frames we could have connected the wire-vertices to either the $A$-vertices or the $B$-vertices. 
    
    Note that the drawings given in \cref{fig:clausegadgetmiddle} and \cref{fig:clausegadgetleft} are exemplary. The distances between the wire gadgets on the wire-line might be different. However, it is easy to see that the same construction will work with slightly different drawings of the vertices that appear in the eight frames, but with the same underlying topological embeddings.
\end{proof}
We obtain the following by combining \cref{corr:oldirection} and \cref{lemma:newdirection}.
\begin{restatable}{theorem}{hardness}\label{thm:hardnessthm}
    Deciding whether a graph admits a planar geometric storyplan is \NP-hard. 
\end{restatable}

\section{Conclusion}
We have shown that not all graphs that admit a planar storyplan also admit a planar geometric storyplan, yet the decision problem remains \NP-hard.
A natural open question is whether the problem remains in \NP, or whether it is complete for some other complexity class, such as $\exists\mathbb{R}$.

For further open problems on storyplans in general, we refer to those given in \cite{BinucciGLLMNS24} and \cite{FialaFLWZ24}.

\bibliography{literature}

\end{document}